\documentclass[11pt, a4paper]{article}
\pdfoutput=1 

\usepackage{fullpage}

\usepackage{mdframed}

\usepackage[utf8]{inputenc}

\usepackage{libertine}

\usepackage{amsfonts}
\usepackage{amsmath}
\usepackage{amssymb}
\usepackage{amsthm}
\usepackage{float}
\usepackage{enumitem}
\usepackage[svgnames]{xcolor}

\usepackage{tikz}  
\usetikzlibrary{arrows}
\usetikzlibrary{patterns,snakes}
\usetikzlibrary{decorations.shapes}
\tikzstyle{overbrace text style}=[font=\tiny, above, pos=.5, yshift=5pt]
\tikzstyle{overbrace style}=[decorate,decoration={brace,raise=5pt,amplitude=3pt}]
\usetikzlibrary{shapes.geometric}

\usepackage{subcaption}

\usepackage{cmap}
\usepackage[breaklinks=true]{hyperref}
\hypersetup{colorlinks={true},urlcolor={blue},linkcolor={DarkBlue},citecolor=[named]{DarkGreen}}

\usepackage[authoryear,square]{natbib}

\usepackage{xspace}
\usepackage{fancyhdr}
\usepackage{tcolorbox}

\usepackage{microtype}
\usepackage[capitalise,nameinlink]{cleveref}

\usepackage{doi}

\usepackage{booktabs} 
\usepackage[ruled,linesnumbered]{algorithm2e} 
\usepackage{setspace}

\SetAlFnt{\small}
\SetAlCapFnt{\small}
\SetAlCapNameFnt{\small}
\SetAlCapHSkip{0pt}
\IncMargin{-\parindent}

\usepackage{authblk}
\usepackage{comment}

\newcommand{\eps}{\ensuremath{\varepsilon}\xspace}

\newcommand{\ical}{\ensuremath{\mathcal{I}}\xspace}
\newcommand{\jcal}{\ensuremath{\mathcal{J}}\xspace}

\newcommand{\calM}{\ensuremath{\mathcal{M}}\xspace}

\newcommand{\cost}{\text{cost}}
\newcommand{\SC}{\text{SC}}
\newcommand{\SW}{\text{SW}}
\newcommand{\MC}{\text{MC}}
\newcommand{\leftmost}{\text{left}}
\newcommand{\rightmost}{\text{right}}

\newcommand{\C}{\mathcal{G}}

\newtheorem{theorem}{Theorem}[section]

\newtheorem{lemma}[theorem]{Lemma}

\theoremstyle{definition}
\newtheorem{definition}[theorem]{Definition}

\usepackage[capitalise,nameinlink]{cleveref}
\crefname{enumi}{}{}

\Crefname{appendix}{Supplement Section}{Supplement Sections}

\allowdisplaybreaks

\title{\bf Truthful Interval Covering}

\author[1]{Argyrios Deligkas}
\author[2]{Aris Filos-Ratsikas}
\author[3]{Alexandros A. Voudouris}

\affil[1]{Royal Holloway University of London, United Kingdom}
\affil[2]{University of Edinburgh, United Kingdom}
\affil[3]{University of Essex, United Kingdom}

\date{}

\begin{document}

\maketitle

\begin{abstract}
We initiate the study of a novel problem in mechanism design without money, which we term {\em Truthful Interval Covering} (TIC). 
An instance of TIC consists of a set of agents each associated with an individual interval on a line, and the objective is to decide where to place a {\em covering interval} to minimize the total social or egalitarian cost of the agents, which is determined by the intersection of this interval with their individual ones. This fundamental problem can model situations of provisioning a public good, such as the use of power generators to prevent or mitigate load shedding in developing countries. In the strategic version of the problem, the agents wish to minimize their individual costs, and might misreport the position and/or length of their intervals to achieve that. Our goal is to design \emph{truthful} mechanisms to prevent such strategic misreports and achieve good approximations to the best possible social or egalitarian cost. We consider the fundamental setting of known intervals with equal lengths and provide tight bounds on the approximation ratios achieved by truthful deterministic mechanisms. For the social cost, we also design a randomized truthful mechanism that outperforms all possible deterministic ones. Finally, we highlight a plethora of natural extensions of our model for future work, as well as some natural limitations of those settings. 
\end{abstract}

\section{Introduction}

We introduce the \emph{Truthful Interval Covering (TIC)} problem, a novel problem in the field of mechanism design without money \citep{PT09}. In this problem, there is a set $N$ of $n$ agents, each of whom is associated with an interval $I_i$ on the line of real numbers. There is also a \emph{covering interval} $C$, which should be placed somewhere on the line. The cost of agent $i \in N$ is a function of the portion of $I_i$ that $C$ covers; in the simplest version of the problem, the cost is just the part of $I_i$ that is not covered by $C$. The goal is to place the interval so as to minimize the \emph{social cost} (total cost of the agents) or the \emph{max cost} (maximum individual agent cost), while taking the incentives of the agents into account. Indeed, agents might misreport information about their intervals (e.g., their position or length) if that would lead to an outcome that is preferable for them. 

The TIC problem captures applications in which a public good is provisioned and shared among a set of participants. 
We provide a few indicative of many examples below.
\begin{itemize}[leftmargin=*]
\item The covering interval could represent the time interval during which a power generator can be operated, and the individual intervals capture the times during which each citizen would like to have access to electricity. The minimum-social cost solution is one that covers as much demand for electricity as possible. This is particularly relevant in developing countries where electricity might be a scarce resource, and can be used to prevent or mitigate the effects of load shedding.\footnote{E.g., see \url{https://en.wikipedia.org/wiki/South_African_energy_crisis}}  
\item The covering interval could correspond to the range of a public WiFi hotspot to be placed in an area with low broadband connectivity, when the agents' intervals are the signal ranges of their devices. 
\item The covering interval could capture the time in which to schedule a university open-day or a job fair, given the preferences of the potential attendees over the different time intervals in the day. 
\item The covering interval could be an express public transportation line connecting parts of a city or intercity network, and the agents express which parts of the route they would like this service to cover.  
\end{itemize}

Despite its fundamental nature, and its resemblance to other classic algorithmic problems like the interval scheduling problem and its variants \citep{kolen2007interval}, the interval covering problem has seemingly not been studied systematically from a purely algorithmic point of view. This can likely be attributed to the fact that the optimal covering can be found in polynomial time via a rather simple algorithm (see \Cref{prop:algorithmic-easy} in \Cref{sec:Preliminaries}). Once we move to a mechanism design regime however, where the incentives of the agents for miresporting come into effect, the problem becomes much more challenging. \emph{Truthful} mechanisms, which eliminate those incentives, are necessarily suboptimal, and resort to approximations. Our goal is to design truthful mechanisms that achieve approximations that are as small as possible, and identify the limitations of such mechanisms via appropriate inapproximability results.  

\subsection{Our Contribution}\label{sec:our-results}
In this paper, we introduce the Truthful Interval Covering (TIC) problem as a novel and interesting problem in mechanism design without money. Our technical contribution is as follows: 
\begin{itemize} 
\item We provide upper and lower bounds on the approximation ratio of truthful mechanisms for the most fundamental version of the problem, where all of the interval lengths are \emph{known} and \emph{equal}, which already turns out to be quite challenging. We start with the social cost objective and deterministic truthful mechanisms, for which, in \Cref{sec:deterministic}, we prove a \emph{tight bound} of $2-2/n$ on the approximation ratio. 
In \Cref{sec:randomized}, we present a simple \emph{randomized, universally truthful} mechanism that achieves an approximation ratio of $5/3$, thus outperforming all deterministic ones. 
In \Cref{sec:max}, we turn our attention to the max cost objective, for which we show a tight approximation ratio of $2$ for deterministic mechanisms, and a lower bound of $2$ for a natural class of randomized mechanisms, thus showing that randomization might not be able to lead to improvements for this objective. 

\item We also consider two natural extensions of the main model in \Cref{sec:extensions}. In the first one, the agent intervals are assumed to be unknown, and thus the agents can misreport their starting positions as well as the lengths of their intervals. We show that no truthful mechanism can achieve a meaningful approximation ratio in terms of both the social and the max cost. In the second extension, we consider the case where the interval lengths are known but might be unequal. We show that a simple mechanism, which places the covering interval at the starting position of the agent with the maximum-length interval is truthful and achieves a linear approximation ratio in terms of the social cost, and an approximation ratio of at most $2$ for the max cost; the latter is best possible when the interval lengths are known (any might be equal or unequal). 
\end{itemize}
In \Cref{sec:open}, we present and discuss several other interesting variants of the main model, which capture a wealth of different possible application domains. 
We believe that there is great potential for follow-up work, and the problem could enjoy similar success as other problems within the research agenda of mechanism design without money, such as truthful facility location \citep{chan2021mechanism,PT09}, truthful resource allocation \citep{krysta2014size,filos2014social,abebe2020truthful}, or impartial selection \citep{alon2011sum,fischer2014optimal,bjelde2017impartial}.

\subsection{Related Work and Discussion}

The research agenda of {\em approximate mechanism design without money} was put forward by \citet{PT09} and aims to capture settings involving selfish participants, in which truthful mechanisms are used to optimize a social objective. These mechanisms are compared, via their approximation ratio, against the performance of the best-possible outcome, which would be achievable if the participants were not selfish. The prototypical problem in this field is that of truthful (or strategyproof) facility location, which, following its inception in \citep{PT09}, has flourished into an extremely fruitful research area, giving rise to a plethora of works on several different variants; see the the comprehensive survey of \citet{chan2021mechanism} for details. 

Our setting is markedly different from facility location, where the cost of an agent depends on the \emph{distance} from the location of the facility. In contrast, in our case, the cost of an agent is a function of how much her associated interval is covered. 
Still, there are some conceptual similarities between the two problems, namely in terms of the truthful mechanisms employed to achieve the approximation guarantees. In particular, similarly to the literature of facility location, we also employ mechanisms that are based around $k$-th ordered statistics (e.g., the median) of the agents' reports. These are in fact not particular to facility location, but more generally centered around the concept of single-peaked preferences \citep{black1948rationale,moulin1980}. Despite this superficial connection, the proofs for the performance of these mechanisms are very much different in the TIC problem; it is worth mentioning that, contrary to the setting of \citep{PT09}, in TIC these mechanisms provably do not admit social-cost minimizing outcomes. 

Another related problem is that of strategyproof activity scheduling studied by \citet{xu2020strategyproof}, in which an activity (represented by an interval) is to be placed on a line based on the preferences of self-interested agents. Despite the superficial similarity, this setting is again notably different from ours; in activity scheduling, each agent reports a single point and her cost is the distance from the closest endpoint of the activity interval. This makes the problem much closer to facility location rather than our covering problem. The work of \citet{bei2022cake} on truthful cake sharing is also related to our paper. In contrast to our model, where the agents report their intervals on the line and have costs for the chosen covering interval, in the cake cutting model of \citeauthor{bei2022cake}, the agents report piecewise uniform utilities over a cake (represented as a fixed-size interval) and the objective is to choose an interval of certain length that they will all share. 

From a purely algorithmic point of view (without incentives), problems related to intervals (like scheduling or coloring) are rather fundamental and included in most textbooks on algorithms, e.g., see \citep{kleinberg2006algorithm,roughgarden2022algorithms}. As we mentioned earlier, the algorithmic variant of TIC admits an easy polynomial-time algorithm, and the problem becomes challenging once studied under the mechanism design regime. 

Finally, we remark that the term ``truthful interval cover'' has been used before in the literature for mechanisms with money for solving a crowdsourcing problem, where agents bid for intervals of tasks that they are willing to get; see~\citep{dayama2015truthful,markakis2022improved}. This model is completely different compared to the one we propose and study here, and, thus, we do not expect any ambiguity between our problem and theirs.   


\section{The Setting}\label{sec:Preliminaries}
In the Truthful Interval Covering (TIC) problem, there is a set $N$ of $n$ \textit{agents}, each associated with an \textit{interval} $I_i=[s_i,t_i]$ on the line of real numbers. There is also a \textit{covering interval} $C=[s,t]$ whose position needs to be determined. We focus on the most fundamental version of the problem, in which the interval lengths are all known and equal, i.e., $|I_i|=|I_j|=|C|$ for any agents $i,j \in N$. Given this, we can assume without loss of generality that $|I_i| = 1$ for any $i \in N$ and $|C|=1$. Let $\ical = (I_1, I_2, \ldots, I_n)$ be the vector of the agents' intervals, to which we refer to as an \emph{instance}. Without loss of generality, we assume that for any two agents $i,j \in N$ with $i<j$, $s_i \leq s_j$, i.e., the intervals in $\ical$ appear in non-decreasing order of left endpoints. Using this, we may refer to an agent $i$ being before or after another agent if $i < j$ or $i > j$, respectively. We will also say that an agent $i$ is before or after a point $x$ if $s_i \leq x$ or $s_i > x$. 

Given a position for the covering interval $C$, the cost of an agent $i \in N$ is the part of her interval $I_i$ that does not overlap with $C$, i.e., 
\begin{equation*}
\cost_i(C) = |I_i \setminus (I_i \cap C)| = 1 - |I_i \cap C|.
\end{equation*}
The {\em social cost} of $C$ is the total cost of all agents:
$$\SC(C) = \sum_i \cost_i(C).$$
The {\em max cost} of $C$ is the maximum cost over all agents:
$$\MC(C) = \max_i \cost_i(C).$$
When the objective (social or max cost) is clear from context, for any instance $\ical$, we will use $O(\ical)$ to denote the covering interval that minimizes the objective for instance for $\ical$; when $\ical$ is also clear from context, we will simply write $O$.

A deterministic {\em mechanism} $\calM$ takes as input an instance $\ical$ and outputs the position of a covering interval $\calM(\ical)$, i.e., the position $s \in \mathbb{R}$ of the left endpoint of the covering interval. We also consider \emph{randomized} mechanisms, which, instead of a single position, output a probability distribution $D_{\calM}(\ical)$ over possible positions of the covering interval. 

The {\em approximation ratio} of $\calM$ in terms of an objective $f \in \{\SC,\MC\}$ is the worst-case ratio of the objective value of the covering interval computed by $\calM$ over the smallest possible objective value achieved by any covering interval, over all instances of the problem: 
\begin{align*}
    \sup_{\ical} \frac{f(\calM(\ical))}{\min_{C} f(C)}.
\end{align*}
For randomized mechanisms, the definition is very similar, with the only difference that the \emph{expected objective value} $\mathbb{E}_{\calM(\ical) \sim D_{\calM}(\ical)}[f(\calM(\ical))]$ appears in the numerator. 

The term ``Truthful'' in the name of the TIC problem comes from the fact that the information about the intervals is not public knowledge, but has to be elicited from the agents. The agents are self-interested entities who might misreport this information if that results in them achieving a smaller cost. In the setting we consider, since the interval lengths are all $1$, the elicited information is the position of the interval of each agent $i \in N$, i.e., the left endpoint $s_i \in \mathbb{R}$ of $I_i$. For simplicity, we say that each agent reports her interval $I_i$ rather than $s_i$. We discuss generalizations to what constitutes elicited information in \Cref{sec:discussion}. 

A mechanism $\calM$ is said to be {\em truthful} if does not incentivize the agents to misreport their intervals, that is, for every agent $i$ and every possible interval $I_i'$ that the agent could report,
\begin{equation}
\label{eq:truthfulness}
\cost_i(\calM(\ical)) \leq \cost_i(\calM(I_i',\ical_{-i}))
\end{equation}
where $\ical_{-i}$ is the vector $\ical$ without the $i$-th coordinate.

For randomized mechanisms, the definition of truthfulness extends to \emph{truthfulness in expectation}, which stipulates that no agent can decrease her expected cost by deviating. In our positive results, we will actually use a stronger truthfulness guarantee called \emph{universal truthfulness}. A mechanism is universally truthful if it is truthful for any realization of truthfulness, i.e., Inequality~\eqref{eq:truthfulness} holds for any $\calM(\ical) \sim D_{\calM}(\ical)$.

Our goal in this paper is to design truthful mechanisms (either deterministic truthful or universally truthful) with approximation ratios as close to $1$ as possible. To achieve this, we focus on the following class of mechanisms, called \emph{$k$-ordered statistics}.

\begin{definition}[$k$-ordered statistic]\label{def:kth-ordered-statistic}
For $k \in [n]$, the $k$-ordered statistic mechanism $\calM$ outputs the interval reported by the $k$-th ordered agent $i$ in instance $\ical$, i.e., $\calM(\ical)=I_i$. 
\end{definition}
For example, for $k=\lfloor n/2 \rfloor$, the $k$-ordered statistic mechanism outputs exactly the interval reported by the median agent. $k$-ordered statistic mechanisms (as well as their convex combinations) are well-known to be truthful in other contexts, e.g., see \cite{dummett1961stability, moulin1980}. We prove that, for similar reasons, any $k$-ordered statistic mechanism is truthful in our setting.

\begin{theorem}\label{lem:kth-ordered-statistics-truthful}
For any $k \in [n]$, the $k$-ordered statistic mechanism is truthful. Furthermore, any convex combination over $k$-ordered statistic mechanisms is universally truthful. 
\end{theorem}

\begin{proof}
Let $\calM$ be the $k$-ordered statistic mechanism outputs the interval of the $k$-th ordered agent $i$. Clearly, agent $i$ has cost $0$ as her interval is completely covered.
Consider the $\ell$-th ordered agent $j$ with $\ell < k$; the case $\ell > k$ is essentially symmetric. 
If $j$ reports any position $s_j' < s_i$, then the outcome of the mechanism will not change. 
If $j$ reports some position $s_j' \geq s_i$, then the mechanism will place the covering interval so that it starts at some position $x \geq s_i$, which cannot decrease the cost of $j$ as it is farther from $j$'s true interval. Hence, $j$ has not incentive to misreport and the mechanism is truthful. 

Finally, observe that, since any $k$-th ordered statistic mechanism is truthful, any convex combination of such mechanisms is universally truthful by definition. 
\end{proof}


Before we proceed with the design of truthful mechanisms, we state and prove the following statement, which establishes that the purely algorithmic version of the problem, without any regard to agent incentives, can be solved in polynomial time with respect to the social cost and the max cost. We refer to this problem as the {\sc Interval Covering Problem}. 

\begin{theorem}\label{prop:algorithmic-easy}
The social cost-minimizing position and the max cost-minimizing position for the covering interval in the {\sc Interval Covering Problem} can be computed in linear time.
\end{theorem}

\begin{proof}
For the social cost, let $S$ be the set of left and right endpoints of the intervals of all agents. We will prove that the social cost-minimizing position of the covering interval starts or ends at a point in $S$. From that, it follows that the optimal covering interval can be computed in linear time by simply checking all the points of $S$. Suppose that there is an instance in which the optimal position of the covering interval does not start or end at some point of $S$. Let $L$ and $R$ be the sets of agents whose rightmost and leftmost endpoints are covered by this interval, respectively; so any agent of $L$ starts before the interval, any agent of $R$ starts after the interval, and all of them have positive intersection with the interval. Clearly, either $|L| \geq |R|$ or $|R|>|L|$. If $|L| \geq |R|$, we can shift the interval towards the left until we meet the rightmost endpoint of one of the agents in $L$. As we do this, the social cost decreases due to increasing the intersection with the agents of $L$ at least as much as it increases due to decreasing the intersection with the agents of $R$, and thus the this new interval must also be optimal. The case $|R| > |L|$ is similar with the only difference that the interval can be moved towards the right. 

For the max cost, the max cost-minimizing position of the covering interval is $(s_\ell+s_r)/2$, where $s_\ell$ and $s_r$ are the starting positions of the leftmost and the rightmost agent intervals. Clearly, if the optimal interval starts at some position smaller than $s_\ell$ or larger than $s_r$, then the interval that starts at exactly $s_\ell$ or exactly $s_r$, respectively, leads to at most the same max cost and is thus optimal as well. Now, suppose that the optimal interval starts at a position $(s_\ell+s_r)/2 + x$ for some $x > 0$; the case $x<0$ is symmetric. Then, the max cost is equal to the cost of the leftmost agent, which is equal to $\min\{ 1, (s_r-s_\ell)+x\}$. Consequently, the cost is minimized for $x=0$.
\end{proof}

\section{Social Cost: Deterministic Mechanisms}\label{sec:deterministic}
We start by showing bounds on the approximation ratio of deterministic truthful mechanisms for the social cost. 
As mentioned in \Cref{sec:our-results}, for this case we obtain a tight bound of $2-2/n$ on the approximation ratio achievable by any such mechanism. The mechanism that achieves this bound is the {\sc Median} mechanism, the $k$-ordered statistic mechanism with $k = \lfloor n/2 \rfloor$ (see \Cref{def:kth-ordered-statistic}). Similar mechanisms (that choose the reported action of the median agents) have played a prominent role in other domains in mechanism design without money \cite{chan2021mechanism}. However, as we mentioned earlier, the nature of our problem is different from those, and hence the proof is also rather different. 

\begin{theorem}\label{thm:median-UB}
The {\sc Median} mechanism achieves an approximation ratio of $2-\frac{2}{n}$ for the TIC problem. 
\end{theorem}

\begin{proof}
Consider an arbitrary instance. 
Let $m$ be the median agent so that $I_m = [s_m,t_m]$ is the unit-size interval that is chosen by the mechanism.
Let $O = [s_o,t_o]$ be the optimal unit-size interval.  
Without loss of generality, we can assume that $s_o \geq s_m$ and that $n$ is even.
Let $x \in [0,1]$ be the length of the intersection $I_m \cap O = [s_o,t_m]$ between the interval chosen by the mechanism and the optimal interval.
Let $L$ be the set of $n/2-1$ agent at the left of $m$, $M$ the set of agents between $m$ and $s_o$, and $R$ the set of agents at the right of $s_o$; note that $|M \cup R| = n/2$. Clearly, the maximum cost of the mechanism is $n$ and the minimum optimal cost is $0$. 
We make the following observations: 
\begin{itemize}[leftmargin=*]
    \item The median agent decreases the cost of the mechanism by $1$ and increases the optimal cost by $1-x$ (the part of the median interval that the optimal solution does not cover).\smallskip
    \item Any agent $i \in L$ such that $|I_i \cap [s_o,t_m]| = 0$ increases the optimal cost by $1$; let $A$ be the set of all such agents. \smallskip
    \item Any agent $i \in L$ such that $|I_i \cap [s_o,t_m]| = x_i > 0$ decreases the cost of the mechanism by at least $1-x + x_i$ and increases the optimal cost by least $1-x_i$ (since the interval of $i$ starts before $s_m$ but reaches $s_o$); let $B$ be the set of all such agents.\smallskip
    \item Any agent $i \in M$ decreases the cost of the mechanism by at least $x$. \smallskip
    \item Any agent $i \in R$ such that $|I_i \cap [s_o,t_m] = x_i > 0$ decreases the cost of the mechanism by some length $x_i \in [0,x]$ and increases the optimal cost by $1-x_i-(1-x)=x-x_i$; let $\Gamma$ be the set of all such agents. \smallskip
    \item Any agent $i \in R$ such that $|I_i \cap [s_o,t_m]| = 0$ increases the optimal cost by at least $x$; let $\Delta$ be the set of all such agents. 
\end{itemize}
Hence, we have
\begin{small}
\begin{align*}
    \SC(I_m) &\leq n-1-|B|(1-x)-\sum_{i \in B}x_i - |M|x - \sum_{i \in \Gamma} x_i\\
    \SC(O) &\geq 1-x + |A| + |B| - \sum_{i \in B} x_i + |\Gamma|x - \sum_{i \in \Gamma} x_i + |\Delta|x. 
\end{align*}
\end{small}
So, the approximation ratio is at most
\begin{align*}
    \frac{n-1-|B|(1-x)-\sum_{i \in B}x_i - |M|x - \sum_{i \in \Gamma} x_i}{1-x + |A| + |B| - \sum_{i \in B} x_i + |\Gamma|x - \sum_{i \in \Gamma} x_i + |\Delta|x}.
\end{align*}
Since the ratio is at least $1$ (by definition), it is an increasing function in terms of $\sum_{i \in B}x_i \leq |B|x$ and it terms of $\sum_{i \in \Gamma}x_i \leq |\Gamma|x$ , and is thus at most
\begin{align*}
    \frac{n-1-|B|-(|M| + |\Gamma|)x }{1-x + |A| + |B|(1-x) + |\Delta|x}.
\end{align*}
Since $|A| + |B| + 1 = n/2$ and $|M|+|\Gamma|+|\Delta|=n/2$, we further obtain 
\begin{align*}
    \frac{n-1-|B|-(n/2)\cdot x + |\Delta|x}{n/2 -x -|B|x + |\Delta|x}.
\end{align*}
This is a decreasing function in terms of $|\Delta| \geq 0$, and is thus at most
\begin{align*}
    \frac{n-1-|B| -(n/2)\cdot x }{n/2 -x -|B|x }.
\end{align*}
This is a decreasing function in terms of $|B| \geq 0$, and is thus at most
\begin{align*}
    \frac{n-1-(n/2)\cdot x}{n/2-x}.
\end{align*}
Finally, this is a decreasing function in terms of $x$ and thus attains its maximum value of $2-2/n$ when $x=0$. 
\end{proof}

Next, we present a lower bound for deterministic truthful mechanisms that matches the upper bound of \Cref{thm:median-UB}.
Before we do so though, we will provide a structural property of any deterministic truthful mechanism. This property will be repeatedly used in order to prove the lower bound. 

\begin{lemma}
\label{lem:det-struct}
Consider a deterministic truthful mechanism $\calM$, an instance $\ical$, and an agent $i$ such that $I_i \cap \calM(\ical) \neq \varnothing$. 
In addition, consider the instance $\ical' = (I'_i, \ical_{-i})$, where $I'_i \cap I_i \cap \calM(\ical) \neq \varnothing$, and let $\calM(\ical')$ be the location of the covering interval in $\ical'$ under mechanism $\calM$. Then, it must hold that $I'_i \cap \calM(\ical') \neq \varnothing$.
\end{lemma}
\begin{proof}
Assume by contradiction that for the instance $\ical'$ it holds that $I'_i \cap I' = \varnothing$. Thus, on instance $\ical'$ agent $i$ incurs a cost of $1$. In this case though, agent $i$ could report $I_i$ as their true interval, which would force the mechanism to locate the covering interval at $\calM(\ical)$. This in turn implies that agent $i$ incurs cost strictly less than $1$, since $I'_i \cap I_i \cap \calM(\ical) \neq \varnothing \Rightarrow I'_i \cap \calM(\ical) \neq \varnothing$.
\end{proof}

\begin{theorem}
\label{thm:2-LB}
Let $\calM$ be any deterministic truthful mechanism. Then the approximation ratio of $\calM$ is at least $2-\frac{2}{n}$.
\end{theorem}

\begin{proof}
Let $\calM$ be any deterministic truthful mechanism. At a high level, the proof will construct a series of instances and will use the truthfulness of $\calM$ to argue about the possible positions of the covering interval on each one of those instances.

The starting point is the following instance $\ical^0$, with two groups of agents: group $\C_0$ contains $\frac{n}{2}$ agents with $I_i=[0,1]$ for all $i \in \C_0$ and group $\C_1$ contains $\frac{n}{2}$ agents with $I_i=[n, n+1]$ for all $i \in \C_1$. 
Without loss of generality, we will assume that on instance $\ical^0$, mechanism $\calM$ locates the covering interval $[a_0,b_0]$ such that it covers some part of the intervals of the agents from cluster $\C_0$; the other case is symmetric. 
Observe here that throughout the proof it is without loss of generality to assume that the covering interval {\em always} covers a strictly positive part of an agent; if this was not the case, then the optimal cost would be $n/2$ while the mechanism would achieve cost of $n$ and thus it would be 2-approximate.
In addition, again without loss of generality, we will assume that $0 \leq a_0 \leq 1$.
Observe that since the covering interval has length 1, then it cannot cover any agent from cluster $\C_1$.

The proof will construct a sequence of {\em families} of instances $\jcal^0, \jcal^1, \jcal^2, \ldots, \jcal^k$, where $k \leq n/2-1$, which will guarantee that, on any instance in any of these families: 
\begin{enumerate}[label=(\alph*)]
\item mechanism $\calM$ {\em cannot} place the covering interval and cover (part of) the cluster of agents located at $[n,n+1]$; 
\item mechanism $\calM$ can move the covering interval only to the right of its position in the previous instance and never to the left; 
\item the maximum approximation ratio that mechanism $\calM$ can achieve will strictly decrease, compared to the previous family.
\end{enumerate}
 
Before we present the formal argument, we define some notation that will make the exposition more clear. 
For every instance $\ical$, let $X(\ical) = \{i \in [n]: I_i \cap \calM(\ical) \neq \varnothing\}$, i.e., the set $X(\ical)$ contains the agents that have a non-empty interscetion with the covering interval $\calM(\ical)$. 

We will prove by induction that for every family $\jcal^k$, with $k \in [\frac{n}{2}-1]$, and every instance $\ical \in \jcal^k$ the following two conditions are satisfied.
 \begin{enumerate}
     \item The left endpoint of $\calM(\ical)$ is in $[k, k+1]$.
     \item On instance $\ical \in \jcal^k$, due to truthfulness, mechanism $\calM$ is able to cover a total mass of at most $n/2-k$, for a social cost of at least $n-(n/2-k)=n/2+k$, while the optimal social cost will remain the same, equal to $n/2$. Hence, when the argument reaches family $\jcal^{\frac{n}{2}-1}$, i.e., when $k=n/2-1$, the social cost of $\calM$ becomes at least $n-1$, leading to an approximation ratio of $2-2/n$. For a visual representation of instance $\ical \in \jcal^{\frac{n}{2}-1}$, see the right-hand side of \Cref{fig:randomized}.
 \end{enumerate}

 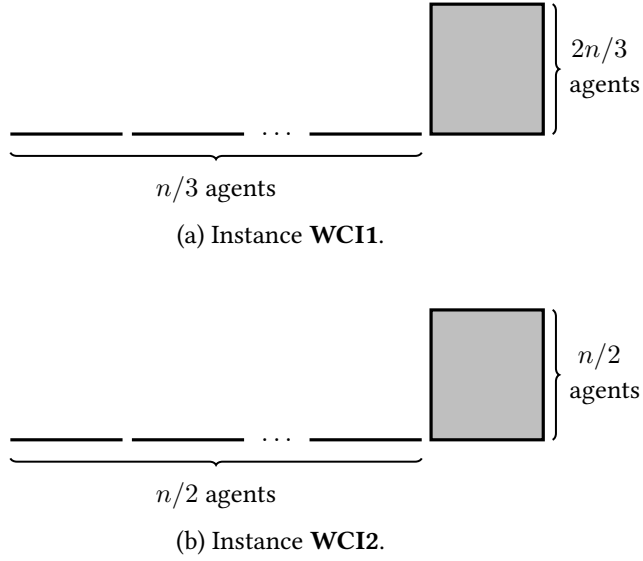
\begin{figure*}[t]
\tikzset{every picture/.style={line width=0.75pt}} 
\tikzset{
    position label/.style={
       below = 3pt,
       text height = 1.5ex,
       text depth = 1ex
    },
   brace/.style={
     decoration={brace, mirror},
     decorate
   }
}
\centering
\begin{subfigure}[t]{0.45\linewidth}
\centering
\begin{tikzpicture}[x=0.7pt,y=0.7pt,yscale=-1,xscale=1]
\draw[very thick] (0,0) -- (60,0); %
\draw[very thick] (65,0) -- (125,0); %
\draw[very thick] (142.5,0) node [inner sep=0.75pt]  [font=\small]  {$\ldots$};
\draw[very thick] (160,0) -- (220,0); %
\draw[very thick, fill=gray, fill opacity=0.5] (225,0) rectangle (285,-70);

\draw [brace] (0,10) -- (220,10);
\draw (110,30) node [inner sep=0.75pt]  [font=\small]  {$n/3$ agents};

\draw [brace] (290,0) -- (290,-70);
\draw (315,-45) node [font=\small]  {{$2n/3$}};
\draw (318,-25) node [font=\small]  {agents};
\end{tikzpicture}
\caption{Instance {\bf WCI1}.}
\label{fig:wci1}
\end{subfigure}
\\[20pt]
\begin{subfigure}[t]{0.45\linewidth}
\begin{tikzpicture}[x=0.7pt,y=0.7pt,yscale=-1,xscale=1]
\draw[very thick] (0,0) -- (60,0); %
\draw[very thick] (65,0) -- (125,0); %
\draw[very thick] (142.5,0) node [inner sep=0.75pt]  [font=\small]  {$\ldots$};
\draw[very thick] (160,0) -- (220,0); %
\draw[very thick, fill=gray, fill opacity=0.5] (225,0) rectangle (285,-70);

\draw [brace] (0,10) -- (220,10);
\draw (110,30) node [inner sep=0.75pt]  [font=\small]  {$n/2$ agents};

\draw [brace] (290,0) -- (290,-70);
\draw (315,-45) node [font=\small]  {{$n/2$}};
\draw (318,-25) node [font=\small]  {agents};
\end{tikzpicture}
\caption{Instance {\bf WCI2}.}
\label{fig:wci2}
\end{subfigure}
\caption{The two worst-case instances for the {\sc Uniform-Statistic} mechanism. In each figure there are some singleton agents and a group of agents whose intervals overal (these are depicted as a shaded rectangle).}
\label{fig:randomized}
\end{figure*}
 
\noindent
Now we are ready to complete the proof. 
Observe that by assumption Conditions 1 and 2 above hold for $\ical^0$, so they hold for the base case, $\jcal^0$, of our induction.
For the induction step, assume that Conditions 1 and 2 hold for instance $\ical \in \jcal^k$, for some $k \leq n/2-1$. 
Hence, we have that the left endpoint of $\calM(\ical)$ is in $[k, k+1]$ and that $\calM(\ical)$ intersects with at most $n/2-k$ agents, formally, $|X(\ical)| \leq n/2-k$.

In what follows, without loss of generality we will assume that $|X(\ical)| > 1$; if $|X(\ical)| \leq 1$, then $\calM(\ical)$ has zero intersection with the intervals of at least $n-1$ agents, and hence has a social cost of at least $n-1$. The optimal cost is $n/2$, and hence $M$ has an approximation ratio of at least $2-2/n$ and we are done. 

Given the above, we can pick a ``rightmost'' agent $i \in X(\ical)$ with interval $I_i=[s_i, t_i]$, i.e., an agent $i$ such that $t_i \geq t_{i'}$ for all $i' \in X(\ical)$. We define instance $\ical' = (I'_i, \ical_{-i})$ as follows, where $\text{left}(I')$ and $\text{right}(I')$ denote the left and right endpoints of interval $I'$ respectively.
\begin{itemize}[leftmargin=*]
    \item If $\leftmost(\calM(\ical)) < t_i < \rightmost(\calM(\ical))$, then $I'_i = [t_i, t_i+1]$.\smallskip
    
    \item If $ t_i \geq \rightmost(\calM(\ical))$, then $I'_i = [\rightmost(\calM(\ical)) - \delta, \rightmost(\calM(\ical))-\delta+1]$, where $\delta>0$ is an arbitrarily small quantity.
\end{itemize}
Observe that in both cases we have that $I_i \cap I'_i \cap \calM(\ical) \neq \varnothing$. Thus, from Lemma~\ref{lem:det-struct} and due to the truthfulness of mechanism $\calM$, it must hold that $I'_i \cap \calM(\ical') \neq \varnothing$. We distinguish between three cases.
\begin{itemize}
    \item $\leftmost(\calM(\ical')) < k+1$ and $|X(\ical')| > 1$. Then, we create a new instance $\ical''$ as before; formally, we set $\ical = \ical'$ and $X(\ical) = X(\ical')$ and we choose an agent from $X(\ical)$ to move.
    \item $\leftmost(\calM(\ical')) < k+1$ and $|X(\ical')| = 1$. Then, as we have argued above, the approximation ratio of mechanism $M$ is $2-2/n$.\smallskip
    \item $\leftmost(\calM(\ical')) \in [k+1, k+2]$. In this case, it holds that $\ical' \in \jcal^{k+1}$. Observe that since we have assumed that $|X(\ical)| > 1$ and that since we have created instance $\ical'$ by moving the ``rightmost'' interval of $X(\ical)$, it should hold that $|X(\ical')| \leq |X(\ical)| - 1 \leq \frac{n}{2}-k-1$, where for the last inequality we have used the induction hypothesis. 
\end{itemize}
This completes the induction step and thus it completes the proof.
 \end{proof}

\section{Social Cost: Randomized Mechanisms} \label{sec:randomized}
In this section, we turn our focus to randomized mechanism for the social cost. 
We present a simple randomized, universally truthful mechanism that achieves an approximation ratio of $5/3$, thus outperforming all deterministic truthful mechanisms.
In particular, we consider the following mechanism, which we coin {\sc Uniform-Statistic}: 

\begin{definition}[{\sc Uniform-Statistic}]
Let $\ell$ be the $(n/3)$-th leftmost agent, $m$ be the median agent, and $r$ be the $(2n/3)$-th leftmost agent. The mechanism places the covering interval at the starting position of each of $\{\ell,m,r\}$ with probability $1/3$. 
\end{definition}

The mechanism is a convex combination of $k$-ordered statistics, and hence by \Cref{lem:kth-ordered-statistics-truthful}, it is universally truthful. What remains is to bound its approximation ratio, established by the following theorem.

\begin{theorem} \label{thm:randomized}
The approximation of {\sc Uniform-Statistic} is at most $5/3$. 
\end{theorem}

\begin{proof}
We will show that any arbitrary instance can be transformed into one of the following two possible worst-cases instances (up to symmetries), by appropriately moving the agents so that the approximation ratio of the mechanism does not decrease.
\begin{itemize}[leftmargin=*]
    \item {\bf WCI1:} The first instance is such that there are (approximately) $2n/3$ agents grouped together, while the remaining $n/3$ agents are all singletons without any intersection with any other agent. The optimal interval completely covers the group of $2n/3$ agents for a social cost of $n/3$. The output of the mechanism coincides with the optimal with probability $2/3$ (due to agents $\ell$ and $m$, or agents $m$ and $r$), and has social cost (approximately) $n$ with probability $1/3$ when it chooses a singleton. So, the approximation ratio is $2/3 + \frac13 \cdot \frac{n}{n/3} = 5/3$. See \Cref{fig:wci1}.

    \item {\bf WCI2:} The second instance is such that there are $n/2$ singleton agents (including $m$) without any intersection with any other agent, while the remaining $n/2$ agents are grouped together. Here, the optimal interval completely covers $n/2$ agents for a social cost of $n/2$. The output of the mechanism coincides with the optimal with probability $1/3$ (due to agent $r$, or agent $\ell$) and has cost (approximately) $n$ with probability $2/3$ when it chooses a singleton. So, the approximation ratio is $1/3 + \frac23 \cdot \frac{n}{n/2} = 5/3$. See \Cref{fig:wci2}.
\end{itemize}

In the initial arbitrary instance, there is an optimal position of the covering interval.
All transformations that we will do will be such that the approximation ratio does not decrease (i.e., does not get better) in terms of keeping the same optimal interval. Of course, the optimal position for the interval might be different for different instance. However, if for two instances $\ical$ and $\ical'$ we show that 
$$\frac{\SC(\calM(\ical))}{\SC(O(\ical))} \leq \frac{\SC(\calM(\ical'))}{\SC(O(\ical))},$$ 
then, since $\SC(O(\ical)) \geq \SC(O(\ical'))$, we can also conclude that
$$\frac{\SC(\calM(\ical))}{\SC(O(\ical))} \leq \frac{\SC(\calM(\ical'))}{\SC(O(\ical'))}.$$ 
Given this, when we refer to the optimal interval $O$ in the rest of the proof, we refer to the optimal position of the interval in the initial instance. 

For any agent $i\in N$, let $x_i \in [0,1]$ denote the length of the intersection between $I_i$ and the optimal interval. 
The next lemma shows that in many cases, we can move agents towards the optimal interval or away from it. 

\begin{lemma}\label{lem:moving-lemma}
Consider any agent $i \not\in \{\ell,m,r\}$. We can obtain an instance with not-better approximation ratio by moving agent $i$
\begin{itemize}
    \item[(a)] towards the optimal interval (up to not changing the order of the agents) if $x_i > 0$
    \item[(b)] towards the optimal interval (up to not changing the order of the agents) if $x_i = 0$ and no $j \in \{\ell,m,r\}$ is inbetween $i$ and the optimal interval;
    \item[(c)] away from the optimal interval (up to not changing the order of the agents) if $x_i = 0$ and the intersection between $i$ and any $j \in \{\ell,m,r\}$ does not increase. 
\end{itemize}
\end{lemma}

\begin{proof}
(a) If $x_i > 0$, then moving the interval of $i$ towards the optimal interval increases the intersection of $i$ with $O$ and thus decreases the optimal social cost. The intersection with any $j \in \{\ell,m,r\}$ may increase or decrease. If the intersection decreases, then the social cost of the mechanism decreases. If the intersection increases, then it increases with the same rate as the intersection of $i$ with $O$, and since each $j \in \{\ell,m,r\}$ is considered with probability $1/3$, the total decrease in the social cost of the mechanism is at most the decrease in the optimal social cost, leading to an instance with not-better approximation ratio. 

(b) If $x_i = 0$ and no $j \in \{\ell,m,r\}$ is inbetween $i$ and the optimal interval, then moving $i$ towards the optimal interval increases the intersection of $i$ with $O$. Until $i$ and $O$ meet, this move does not increasing the intersection between $i$ with any $j \in \{\ell,m,r\}$. After $i$ and $O$ meet, this move can increase the intersection between $i$ and $j$ at most as much as the increase of the intersection between $i$ and $O$. Overall, this leads to a new instance with not-better approximation ratio.

(c) If $x_i=0$ and there is some $j \in \{\ell,m,r\}$ is inbetween $i$ and the optimal interval, then moving the interval of $i$ away from the optimal interval does not change the contribution of $i$ to the optimal social cost, which is the maximum possible. If this move does not increase the intersection with the interval of any $j \in \{\ell,m,r\}$, then the social cost of the mechanism does not decrease (it may increase), and thus the new instance has not-better approximation ratio. 
\end{proof}

We will also use the following observation. 

\begin{lemma}\label{lem:ell-or-r-intersection}
In a worst-case instance, $x_\ell > 0$ or $x_r > 0$. 
\end{lemma}

\begin{proof}
If $x_\ell = 0$ and $x_r = 0$, then the optimal interval has no intersection with at least $2n/3$ agents. Since the social cost of the mechanism is at most $n$, this would lead to an approximation ratio of at most $3/2$.  
\end{proof}

We now partition the agents in $N \setminus \{\ell,m,r\}$ into four sets according to the order of their starting positions as follows:
\begin{itemize}
    \item $A$ includes the agents before $\ell$; 
    \item $B$ includes the agents inbetween $\ell$ and $m$;
    \item $C$ includes the agents inbetween $m$ and $r$;
    \item $D$ includes the agents after $r$. 
\end{itemize}
Note that $|A| = |D|= n/3$ and $|B| = |C|=n/6$. 
Without loss of generality, we can assume that the optimal interval starts weakly to the right of the starting position $s_m$ of $m$, i.e., $s_m \geq s_o$. We will consider two main cases depending on the relative positions of the optimal interval and the interval of agent $r$, i.e., whether $s_r \leq s_o$ or $s_o < s_r$. In each case, we will distinguish between further subcases depending on whether the optimal interval intersects with the intervals of $\ell$, $m$, and $r$. 

\bigskip
\noindent
{\bf Case 1: $s_\ell \leq s_m \leq s_r \leq s_o$.} \ \\
We first argue that the interval of any agent $i\in D$ coincides with the optimal interval in a worst-case instance. By Lemma~\ref{lem:ell-or-r-intersection}, since $s_o \geq s_r$, it has to be the case that $x_r > 0$ and also there is no $j \in \{\ell,m,r\}$ inbetween any agent $i \in D$ and $O$. Therefore, by Lemma~\ref{lem:moving-lemma}, moving each agent $i \in D$ (with $x_i>0$ or $x_i = 0$) to the optimal interval leads to an instance with not-better approximation ratio.

\bigskip
\noindent
{\bf Case 1.1: $x_\ell > 0$.}
Due to the order of the starting positions of $\ell$, $m$, $r$, and $O$, $x_\ell > 0$ implies that $x_i > 0$ for every $i \in N \setminus A$. By Lemma~\ref{lem:moving-lemma}, we can move all the agents of $B \cup C$ towards the optimal as long as the order of the agents does not change, which means that all agents of $B$ can be moved to coincide with $m$ and all agents of $C$ can be moved to coincide with $r$. Again by Lemma~\ref{lem:moving-lemma}, we can move any agent $i \in A$ with $x_i > 0$ to coincide with $\ell$ and we can move any agent $i \in A$ with $x_i = 0$ away from the optimal towards the left as a singleton. Let $A_\ell \subseteq A$ be the subset of agents of $A$ that can be moved to coincide with $\ell$. We now make the following transformations:
\begin{itemize}[leftmargin=*]
    \item All agents of $C \cup \{r\}$ can be moved to $O$. 
    When this is done, the optimal social cost decreases by $\frac{n}{6}(1-x_r)$. The social cost of the mechanism does not decrease at all due to $m$ or $\ell$ since the agents of $C \cup \{r\}$ move away from $m$ and $\ell$. The social cost of the mechanism decreases by $\frac13 \cdot \frac{n}{3}(1-x_r)$ due to the new intersection of $r$ with the agents of $D$ which are all exactly at $O$. Hence, the optimal social cost decreases by more than what the social cost of the mechanism can decrease, and thus the approximation ratio of the new instance is not smaller. \medskip

    \item All agents of $A_\ell \cup \{\ell\}$ can be moved away from $O$ as singletons. 
    When this is done, the optimal social cost increases by $|A_\ell|x_\ell$. 
    The social cost of the mechanism increases by at least $\frac13 |A_\ell|$ due to $\ell$ (who covers completely the agents of $A_\ell$ and will have no intersection with them after the move), at least $\frac13 |A_\ell|x_\ell$ due to $m$ (who intersects by at least $x_\ell$ with any agent of $A_\ell$ since $s_m \leq s_o$), and $\frac13 |A_\ell|x_\ell$ due to $r$ (who coincides with $O$). So, the overall increase of the social cost of the mechanism is at least as much of the increase of the optimal social cost, leading to a new instance with approximation ratio that is at least as large.    
\end{itemize}
Given the above properties, the optimal social cost is 
$$\frac{n}{3} + \frac{n}{6}(1-x_m),$$
whereas the social cost of the mechanism is 
$$\frac13 n + \frac13 \left(\frac{n}{3} + \frac{n}{2}(1-x_m)\right) + \frac13 \left(\frac{n}{3} + \frac{n}{6}(1-x_m) \right),$$
leading to an approximation ratio of 
$$\frac13 + \frac{1 + \frac13 + \frac12 (1-x_m) }{1 + \frac12 (1-x_m)} = \frac43 + \frac13 \cdot \frac{2}{3-x_m} \leq \frac53,$$
since $x_m \geq 1$. This essentially implies that all agents of $B \cup \{m\}$ can be moved to the optimal interval, leading to {\bf WCI1}.

\bigskip
\noindent
{\bf Case 1.2: $x_\ell = 0$.}
Note here that the case where $x_m > 0$ is already covered by Case 1.1: There, we concluded that the worst-case is when $x_\ell$ is essentially $0$. So, it suffices to consider the case where $x_m = 0$. By Lemma~\ref{lem:moving-lemma}, we can move all agents of $A \cup \{\ell\}\cup B \cup \{m\}$ away from the optimal interval to the left as singletons. We can also move every agent $i \in C$ with $x_i > 0$ to $r$, and every agent $i \in C$ with $x_i = 0$ as singleton to the left. 
Let $C_r \subseteq C \cup \{r\}$ be the set of agents that coincide with $r$, and thus have intersection $x_r$ with $O$.
The optimal social cost is 
\begin{align*}
\frac{n}{2} + |C \setminus C_r| + |C_r|(1-x_r), 
\end{align*}
whereas the social cost of the mechanism is
\begin{align*}
\frac23 n + \frac13 \left( \frac{n}{2} + |C \setminus C_r| + \frac{n}{3}(1-x_r) \right), 
\end{align*}
leading to an approximation ratio of 
\begin{align*}
    \frac{ \frac23 n + \frac13 \left( \frac{n}{2} + |C \setminus C_r| + \frac{n}{3}(1-x_r) \right) }
    {\frac{n}{2} + |C \setminus C_r| + |C_r|(1-x_r)}.
\end{align*}
This is a monotone function of $x_r$ and hence it attains its maximum value either for $x_r = 0$ or $x_r=1$. It is not hard to see that for $x_r = 0$ the bound would be at most $6/5$ (as then the optimal social cost would be $6n/5$), whereas for $x_r = 1$, since $|C \setminus C_r| \geq 0$, the approximation ratio is
\begin{align*}
\frac{ \frac23 n + \frac13 \left( \frac{n}{2} + |C \setminus C_r| \right) }
    {\frac{n}{2} + |C \setminus C_r|} \leq 
\frac{ \frac23 n + \frac13 \frac{n}{2} }
    {\frac{n}{2}}
    = 5/3.
\end{align*}
Observe that this case has led to {\bf WCI2}.

\bigskip
\noindent
{\bf Case 2: $s_\ell \leq s_m \leq s_o < s_r$.}
We now first argue that the interval of any agent $i\in C$ coincides with the optimal interval in a worst-case instance. Observe that there is no $j \in \{\ell,m,r\}$ inbetween any agent $i \in C$ and $O$. Hence, by Lemma~\ref{lem:moving-lemma}, moving each agent $i \in C$ (with $x_i > 0$ or $x_i = 0$) to the optimal interval leads to an instance with not-better approximation ratio.

\bigskip
\noindent
{\bf Case 2.1: $x_\ell > 0$ and $x_r > 0$.}
Due to the order of the starting positions of $\ell$, $m$, and $r$, we have that $x_i > 0$ for every $i \in B \cup \{m\}$.
By Lemma~\ref{lem:moving-lemma}, 
we can move any agent $i\in A$ with $x_i > 0$ to $\ell$,
any agent $i \in A$ with $x_i = 0$ as a singleton to the left,
all agents of $B$ to $m$, 
any agent $i \in D$ with $x_i > 0$ to $r$, 
and any agent $i \in D$ with $x_i = 0$ as a singleton to the right. 
Let $A_\ell \subseteq A$ and $D_r \subseteq D$ be the subsets of $A$ and $D$ that are moved to $\ell$ and $r$, respectively. 
Given all these, we have the following.
\begin{itemize}
    \item None of $\ell$, $m$, and $r$ cover the agents of $A\setminus A_\ell$ and of $D\setminus D_r$.
    \item Agent $\ell$ has intersection at least $x_\ell$ with any agent of $B \cup \{m\} \cup C$.
    \item Agent $m$ covers completely the agents of $B$, 
    has intersection at least $x_\ell$ with any agent of $A_\ell \cup \{\ell\}$, 
    intersection of length $x_m$ with the agents of $C$, 
    and intersection $\max\{0,x_m+x_r-1\}$ with the agents of $D_r \cup \{r\}$.
    \item Agent $r$ covers completely the agents of $D_r$, 
    has intersection $\max\{0,x_m+x_r-1\}$ with the agents of $B \cup \{m\}$, 
    and intersection $x_r$ with the agents of $C$.  
\end{itemize}
Putting everything together, the social cost of the mechanism is at most
\begin{align*} 
&|A\setminus A_\ell| + |D\setminus D_r| + \frac13 \cdot \bigg( \frac{n}{3} + |A_\ell| \bigg) (1-x_\ell) \\  
&+\frac13 \cdot \frac{n}{6}(2-x_m-x_r) \\
&+ \frac13 \cdot (|B| + |D_r|)(1-\max\{0,x_m+x_r-1\}) \bigg).
\end{align*}
The optimal social cost is
\begin{align*}
&|A\setminus A_\ell| + |A_\ell|(1-x_\ell) \\
&+ \frac{n}{6}(1-x_m) + |D_r|(1-x_r) + |D\setminus D_r|.
\end{align*}
The implied upper bound on the approximation ratio is a multivariate function of $x_\ell$, $x_r$, $x_m$ that is monotone in each of these variables. Consequently, its maximum value is attained at the extreme values of these variables, i.e., when $x_\ell, x_r, x_m \in \{0,1\}$, subject to the constraint $x_\ell \leq x_m$. It is not hard to verify that the combinations $(0,1,1)$, $(0,0,1)$, $(1,1,0)$ are the worst and all lead to the upper bound of $5/3$ (with an extra step of optimization in terms of sizes of $A_\ell$ and $D_r$ once the values of $x_\ell$ and $x_r$ have been settled) and that they correspond to symmetric versions of \textbf{WCI1} and \textbf{WCI2}.  

\bigskip
\noindent
{\bf Case 2.2: $x_\ell = 0$.} By Lemma~\ref{lem:ell-or-r-intersection}, it has to be the case that $x_r > 0$. Clearly, by Lemma~\ref{lem:moving-lemma}, we can move all agents of $A \cup \{\ell\}$ as singletons away from the optimal to the left. 
We consider two subcases now.

\bigskip
\noindent
{\bf Case 2.2.1: $x_m = 0$.}
Again by Lemma~\ref{lem:moving-lemma}, 
we can move all agents of $B \cup \{m\}$ as singletons away from the optimal interval to the left, 
any agent $i \in D$ with $x_i > 0$ to $r$, 
and any agent $i \in D$ with $x_i = 0$ as a singleton to the right. 
Let $D_r \subseteq D$ be the subset of $D$ that is moved to $r$.
The social cost of the mechanism is
\begin{align*}
    \frac23 n + \frac13 \bigg( \frac{n}{2} + |D\setminus D_r| + \frac{n}{6}(1-x_r) \bigg).
\end{align*}
The optimal social cost is
\begin{align*}
    \frac{n}{2} + |D\setminus D_r| + |D_r|(1-x_r).
\end{align*}
The approximation ratio is a monotone function in terms of $x_r$ and attains its maximum value for either $x_r=0$ or $x_r=1$. For $x_r=0$, the approximation ratio would be at most $6/5$ (since then the optimal cost would be at least $5n/6$), whereas for $x_r=1$, the approximation ratio is 
\begin{align*}
\frac{\frac23 n + \frac{1}{3}\cdot \frac{n}{2} + \frac13 |D\setminus D_r|}{\frac{n}{2} + |D\setminus D_r|} \leq \frac{5}{3},
\end{align*}
where the inequality follows since $|D \setminus D_r| \geq 0$. Observe that this is instance {\bf WCI2}. 

\bigskip
\noindent
{\bf Case 2.2.2: $x_m > 0$.}
By Lemma~\ref{lem:moving-lemma}, 
we can move any agent $i \in B$ with $x_i > 0$ to $m$, 
any agent $i \in B$ with $x_i =0$ as a singleton to the left,
any agent $i \in D$ with $x_i > 0$ to $r$, 
and any agent $i \in D$ with $x_i = 0$ as a singleton to the right. 
Let $B_m \subseteq B$ and $D_r \subseteq D$ be the subsets of $B$ and $D$ that are moved to $m$ and $r$, respectively.
The social cost of the mechanism is at least
\begin{small}
\begin{align*}
&\frac{n}{3} + \frac13 \bigg( \frac{n}{3} + |B\setminus B_m| + \frac{n}{6}(1-x_m) + |D\setminus D_r| + |D_r|(1-x_r) \bigg) \\
& + \frac13 \bigg( \frac{n}{3} + |B\setminus B_m| + |B_m|(1-x_m) + \frac{n}{6}(1-x_r) + |D\setminus D_r| \bigg).
\end{align*}
\end{small}
The optimal social cost is
\begin{align*}
\frac{n}{3} + |B\setminus B_m| + |B_r|(1-x_m) + |D\setminus D_r| + |D_r|(1-x_r).
\end{align*}
Again, the approximation ratio is a monotone function in each of $x_m$ and $x_r$, and so it attains its maximum value for $x_m, x_r \in \{0,1\}$. It is not hard to verify that the worst are the combinations $(1,1)$ and $(0,1)$, leading to {\bf WCI1} and {\bf WCI2}, for which the bound is $5/3$.

\end{proof}

We complement the aforementioned positive result with a lower bound of $3/2$ on the approximation ratio of any randomized truthful mechanism that is a convex combination of $k$-ordered statistic mechanisms. 

\begin{theorem} \label{thm:randomized-lower}
The approximation ratio of any convex combination of $k$-ordered statistic mechanisms is at least $3/2$.
\end{theorem}

\begin{proof}
Let $p$ be the total probability with which any of the first $n/2$ agents is chosen; hence, $1-p$ is the total probability with which any of the remaining $n/2$ is chosen. Without loss of generality, $p \geq 1/2$. Now, consider an instance in which the first $n/2$ agents are singletons, whereas the other $n/2$ agents are all grouped together. The optimal social cost is exactly $n/2$, while the expected social cost of the mechanism is $p \cdot n + (1-p)\cdot n/2$, since with probability $p$ we choose some of the first $n/2$ agents leading to social cost $n$ and with probability $1-p$ we choose some the last $n/2$ agents leading to the optimal social cost of $n/2$. Therefore, the approximation ratio is $1+p \geq 3/2$. 
\end{proof}


\section{Max Cost} \label{sec:max}
In this section, we focus on the max cost and show that the best possible approximation ratio of any deterministic mechanism is $2$, and this is achieved by any $k$-ordered statistic mechanism. 

\begin{theorem} \label{max:upper-2}
For the max cost, the approximation ratio of any $k$-ordered statistic mechanism is at most $2$.
\end{theorem}

\begin{proof}
Consider an arbitrary instance and any $k$-Statistic mechanism. The max cost of the solution computed by the mechanism as well that of the optimal solution depend on how close the intervals of the leftmost and rightmost agents are. We consider the following cases:
\begin{itemize}
    \item If the intervals of the leftmost and the rightmost agents are disjoint and the distance between them is at least $1$, then the max cost of the mechanism and the optimal max cost are both $1$; hence, the mechanism is optimal.

    \item If the intervals of the leftmost and the rightmost agents are disjoint and the distance between them is equal to $1-x$ for some $x \in (0,1)$, then the optimal interval can cover $x/2$ of each of these two agents (and any other agent inbetween them), leading to an optimal max cost of $1-x/2 \geq 1/2$. Since the max cost of the mechanism is again $1$, the approximation ratio is at most $2$ in this case.

    \item If the interval of the leftmost and the rightmost agents have an overlap of $x < 1$, then the max cost of the mechanism is $1-x$. The optimal solution can cover $x+y$ from each of the two agents, where $y$ is such that $x+2y=1 \Leftrightarrow y = \frac{1-x}{2}$, leading to an optimal max cost of $1-x-y=\frac{1-x}{2}$. Hence, the approximation ratio is $2$.
\end{itemize}  
Overall, in any case, the approximation ratio is at most $2$.
\end{proof}

Next, we show that there is no better deterministic truthful mechanism. 

\begin{theorem} \label{max:lower-deterministic-2}
For the max cost, for any $k> 0$ the approximation ratio of any deterministic truthful mechanism is at least $2-\frac{1}{k}$.
\end{theorem}

\begin{proof}
In order to prove our theorem we will produce a sequence of instances with two agents -- Left agent and Right agent -- such that the approximation ratio of any deterministic truthful mechanism will be monotonically increasing and tend to 2.  
The high level idea is that at every iteration of the sequence, either Left agent moves $\eps$ to the left, or Right agent moves $\eps$ to the right; at the same time though, due to truthfulness the mechanism can either (a) follow the agent that moves and lose a ``large'' fraction of the other agent that optimal solution covers, or (b) do not follow the agent that moves and thus lose a fraction of the agent that is covered by the optimal solution.

Formally, we will prove our claim by induction where at instance $\ical_k$ the optimal solution has cost $\frac{k\cdot \eps}{2}$, while any deterministic truthful mechanism achieves cost at least $(k-1)\cdot \eps + \frac{\eps}{2}$.
The initial instance $\ical_1$ consists of two agents -- Left agent with interval $I^l_1$ and Right agent with interval $I^r_1$ -- whose intervals have overlap $1-2\eps$, i.e. $|I^l_1 \cap I^r_1| = 1-2\eps$ and thus $|I^l_1 \cup I^r_1| = 1+2\eps$. 
Let $S_1$ be the solution some deterministic mechanism chooses and assume that $S_1$ is better than 2-approximate. Observe that at least $\frac{\eps}{2}$ from the interval of some agent must be uncovered; formally either $|I_1^l \cap S_1| \leq 1- \frac{\eps}{2}$ or $|I_1^r \cap S_1| \leq 1- \frac{\eps}{2}$. 
In order to create $\ical_2$ we pick the agent that is covered the least by the mechanism and move his interval by \eps; to keep notation simple, for every interval $I = [a,b]$ we denote $I +\eps = [a+\eps,b+\eps]$ and $I -\eps = [a-\eps,b-\eps]$.
\begin{itemize}
    \item If $|I_1^l \cap S_1| \leq 1- \frac{\eps}{2}$, then we create $\ical_2$ by moving Left agent by \eps to the left, i.e., we set $I^l_2 :=  I^l_1-\eps$, and $I^r_2 :=  I^r_1$. 
    The optimal solution achieves cost $\eps$.
    Let $S_2$ be the solution of the mechanism on $\ical_2$. 
    Observe that due to truthfulness, it must hold that $|I_1^l \cap S_2| \leq 1- \frac{\eps}{2}$; otherwise Left agent could misreport and decrease his cost. 
    Hence, there are two possibilities for $S_2$. 
    (a) it moves to left of $S_1$. Then, since $|I_1^l \cap S_2| \leq 1- \frac{\eps}{2}$, it must hold that $|I_2^r \cap S_2| \leq 1- \frac{3\eps}{2}$.
    (b) it moves (weakly) to the right of $S_1$. Then, it is true that $|I_2^l \cap S_2| \leq 1- \frac{3\eps}{2}$. in any case the mechanism achieves approximation $\frac{3\eps}{2}/\eps = \frac{3}{2}$.
    \item Else, we create $\ical_2$ by moving Right agent by \eps to the right, i.e., $I^r_2 := I^r_1 + \eps$, and $I^l_2 := I^l_1$. Then, the optimal cost is $\eps$ and due to truthfulness, it must hold that $|I^r_1 \cap S_2| \leq 1-\frac{\eps}{2}$. Using identical arguments as before we can prove that no matter how the mechanism will choose $S_2$, the approximation will be at least $3/2$.
\end{itemize}

For the induction step, assume that we have instance $\ical_k$ where: $|I_k^l \cap I_k^r| = 1-2k\cdot \eps$ and that solution $S_k$ of the truthful deterministic mechanism satisfies either 
$|I_k^l \cap S_k| \leq 1- (k-1)\cdot \eps - \frac{\eps}{2}$ or 
$|I_k^r \cap S_k| \leq 1- (k-1)\cdot \eps - \frac{\eps}{2}$; this can be ensured by creating $\ical_k$  form $\ical_{k-1}$ by continuously moving the interval of an agent until one of the two conditions is satisfied, since if the mechanism does not achieve any of the above then it will achieve approximation ratio strictly larger than 2.
Observe that the cost of the optimal solution is $\frac{k\cdot\eps}{2}$.
Without loss of generality, assume that $|I_k^l \cap S_k| \leq 1- (k-1)\cdot \eps - \frac{\eps}{2}$. 

Now we will explain how to get instance $\ical_{k+1}$. 
Without loss of generality, assume that $|I_k^l \cap S_k| \leq 1- (k-1)\cdot \eps - \frac{\eps}{2}$. We will move left agent $\eps$ to the left. 
Hence, we set $I_{k+1}^l := I_k^l -\eps$ and $I_{k+1}^r := I_k^r$. 
This means that $|I_{k+1}^l \cap I_{k+1}^r| = 1-2(k+1)\cdot \eps$. So, the cost of the optimal solution is $\frac{(k+1)\cdot\eps}{2}$.
Let $S_{k+1}$ be the solution that the deterministic mechanism chooses for $\ical_{k+1}$.
Due to truthfulness, it must be true that $|I_k^l \cap S_{k+1}| \leq 1- (k-1)\cdot \eps - \frac{\eps}{2}$. 
Using the above and the exact same arguments as in the base case, it must be true that either 
$|I_{k+1}^l \cap S_{k+1}| \leq 1- k\cdot \eps - \frac{\eps}{2}$ (when the mechanism does not move the solution to the left), or 
$|I_{k+1}^r \cap S_{k+1}| \leq 1- k\cdot \eps - \frac{\eps}{2}$ (when the mechanism does move the solution to the left). 
Thus, in any case, the approximation ratio of the mechanism is bounded by $\frac{k\cdot\eps+\frac{\eps}{2}}{\frac{(k+1)\cdot \eps}{2}} = 2- \frac{1}{k}$.
\end{proof}

Our next result shows that no randomized mechanism that is a convex combination of $k$-ordered mechanism can achieved an approximation ratio better than $2$.

\begin{theorem} \label{max:lower-randomized-class-2}
For the max cost, the approximation ratio of any mechanism that is a convex combination of $k$-ordered mechanisms is at least $2$.
\end{theorem}

\begin{proof}
Consider any mechanism that is a convex combination of $k$-ordered statistic mechanisms, and the following instance: Half of the agents have the interval $[0,1]$ and the other half of the agents have the interval $[1,2]$. For this instance, the interval is either places at $0$ (with total probability equal to the probability assigned to all $(k,0)$-mechanisms with $k \leq n/2$) or $1$ (with total probability equal to the probability assigned to all $(k,0)$-mechanisms with $k > n/2$). Hence, the expected max cost is $1$. On the other hand, the optimal solution is to place the interval at $1/2$ to cover half of the interval of each agent for a max cost of $1/2$, thus leading to an approximation ratio of $2$.
\end{proof}

\section{Two Natural Extensions}\label{sec:extensions}
Our results so far concerned the case of known intervals of equal length and settled the problem for deterministic truthful mechanisms, while the best possible approximation ratio for randomized mechanisms is still to be determined. We now present two very natural extensions of the main model which could be better fitted to several of the potential applications of the problem. In particular, we first consider the setting where the lengths of the agent intervals are private information; for this model, it turns out that meaningful approximations and truthfulness are incompatible. Second, we consider the case where the lengths of the agent interval are known but unequal; for this, we show that a finite approximation ratio is possible for the social cost, and we also identify the best possible mechanism for the max cost. 


\subsection{Unknown interval lengths} 
In general, it seems natural to assume that the length of the agent intervals, as well as their positions, could constitute reported information. In this case, however, we prove the following stark impossibilities for the social cost and the max cost.

\begin{theorem}\label{thm:unknown-lengths}
For the social cost, when the lengths of the intervals are unknown, the approximation ratio of any randomized truthful in expectation mechanism is $\Omega(1/\eps)$, for any $\varepsilon \in (0,1)$.
\end{theorem}

\begin{proof}
To begin with, we will prove the lower bound for any deterministic mechanism $\calM$ and then explain how we can suitably augment the idea in order to get the bound for randomized mechanisms. Consider an instance $\ical$ with two agents, where the leftmost agent is associated with the interval $I_1=[0,1]$ and the rightmost agent is associated with the interval $I_2=[3, 3+\eps]$. In addition, assume that the covering interval $C$ has length $1$; clearly $C$ can have a non-zero intersection with at most one of the agents' intervals. 
Given $\ical$, $\calM$ must locate $C$ such that it intersects $[0,1]$; otherwise, its approximation ratio would be larger than $\frac{1}{\eps}$.
In fact, it has to be the case that $|C \cap [0,1]| > \eps$.
Without loss of generality, we can assume that $0 \in C$; that is, $C$ covers the left-agent from the left. 
Now, consider the instance $\ical'$, where the left-agent has the interval $[0,\eps^2]$ while the right-agent still has the interval $[3, 3+\eps]$. Here, the optimal cost is $\eps^2$, achieved when the right-agent is completely covered.
Observe that, since $|C \cap [0,1]| > \eps$ and $0 \in C$, $[0,\eps^2]$ is contained in the covering interval of the mechanism for instance $\ical$. We argue that, due to truthfulness, $\calM$ has to locate $C$ at $\ical'$ such that $[0,\eps^2] \subset C$ as well; otherwise, the left-agent could declare $[0,1]$ as her interval, lead to $\ical$, and thus get cost $0$.
Hence, when given instance $\ical'$ as input, mechanism $\calM$ achieves cost $\eps$, and its approximation ratio is $\frac{\eps}{\eps^2} = \frac{1}{\eps}$.

Next, we show an asymptotic bound for randomizes mechanisms. 
Let $\calM$ be any randomized, truthful in expectation mechanism. We will prove that $\calM$ is $\Omega(\frac{1}{\eps})$-approximate, for any $\eps \in (0,1)$. 
Our starting point is instance $\ical$ presented above. Observe that the minimum social cost on $\ical$ is $\eps$. 
Let $p$ be the total probability with which $M$ chooses a covering interval that has non-zero intersection with $[0,1]$ (observe that this is not necessarily a single interval, but rather a collection of intervals $j$, each chosen with a probability $p_j$, with $\sum_j p_j =p$).
If $p<\frac{1}{2}$, the expected social cost of $M$ is at least $1/2$ (since with probability at least $1/2$ it has zero intersection with $[0,1]$), and thus the approximation ratio is at least $1/2\eps$. So, let us focus on the case where $p \geq 1/2$.
Observe that at least one of the intervals $[0,\eps]$ and $[1-\eps, 1]$ has to be covered with probability at least $p/2 \geq 1/4$. 
To see this, observe that the covering interval {\em cannot} intersect with the ``inner'' interval $[\eps,1-\eps]$ without covering at least one of the two above-mentioned intervals. So, if both intervals are covered with probability smaller than $1/4$ each, we get that the cost of $\calM$ is at least $1+\eps$ and thus its approximation ratio is at least $1/\eps$. Without loss of generality assume that the interval $[0,\eps]$ is covered with probability at least $1/4$.
Now, consider the instance $\ical'$ where the leftmost agent is associated with the interval $[0,\eps^2]$. Clearly, the optimal cost for this instance is $\eps^2$.
We claim that the mechanism should cover the interval $[0,\eps^2]$ with probability at least $1/4$. To see this, observe that if this was not the case, then the leftmost agent could declare the interval $[0,1]$ instead and decrease her cost.
Hence, given instance $\ical'$, the cost of $\calM$ is at least $\frac{1}{4}\eps + \frac{1}{4}\eps^2$ which implies that the approximation ratio of $\calM$ for $\ical'$ is at least $1/4\eps$. This concludes the proof. 
\end{proof}

\begin{theorem} \label{max:unknown-lengths}
For the max cost, when the lengths of the intervals are unknown, the approximation ratio of randomized truthful in expectation mechanism is $\Omega(1/\varepsilon)$, for any $\varepsilon \in (0,1)$.
\end{theorem}

\begin{proof}
We will show the bound for all mechanisms using the following two instances: 
Instance $I$ consists of two agents that are associated with the intervals $[0,1]$ and $[\varepsilon,1+\varepsilon]$. 
Instance $J$ is the same as $I$ with the difference that the first agent now has the interval $[0,\varepsilon^2]$.
We first present the argument for deterministic mechanisms, and then for randomized ones. 

\medskip
\noindent
{\bf Deterministic mechanisms.} Clearly, the optimal max cost in $I$ is $\varepsilon$, achieved by completely covering one of the agents. Let $C$ be the interval chosen by an arbitrary deterministic mechanism when given $I$ as input. The mechanism must cover at least a fraction $1/2$ of the intervals of both agents since otherwise the approximation ratio would be at least $1/(2\varepsilon)$. Since the mechanism covers at least $1/2$ of each agent and $|C|=1$, it must be the case that either $[0,\varepsilon] \subset C$ or $[1,1+\varepsilon] \subset C$. Without loss of generality, suppose that $[0,\varepsilon] \subset C$.

Now consider instance $J$. The mechanism must place the interval $C$ in $J$ such that $[0,\varepsilon^2] \subset C$; otherwise, the first agent would have incentive to misreport her interval as $[0,1]$, thus leading to instance $I$, in which the mechanisms places the interval so that the entire interval $[0,\varepsilon^2]$ of the agent is covered. Consequently, the max cost of the mechanism is at least $\varepsilon$, since it is not possible to cover more that $1-\varepsilon$ of the interval of the second agent. However, the optimal max cost in $J$ is $\varepsilon^2$ by covering entirely the second agent, thus leading to an approximation ratio of at least $\varepsilon/\varepsilon^2 = 1/\varepsilon$. 

\medskip
\noindent 
{\bf Randomized mechanisms.} 
Consider again instance $I$ with optimal max cost $\varepsilon$, and an arbitrary randomized mechanism. 
Suppose that with probability $p \geq 1/2$ mechanism covers at most a fraction $1/2$ of the intervals of both agents; then, the expected max cost of the mechanism would be at least $1/4$, leading to an approximation ratio of $1/(4\varepsilon)$. Hence, with probability at least $p\geq 1/2$, the mechanism covers at least a fraction $1/2$ of the intervals of both agents. When it does so, the mechanism covers either $[0,\varepsilon]$ or $[1,1+\varepsilon]$ (since it cannot cover both at the same time). Without loss of generality, suppose that the mechanism covers $[0,\varepsilon]$ with probability at least $p/2 \geq 1/4$.

Now, when given instance $J$ as input, the mechanism must cover the interval $[0,\varepsilon^2]$ with probability at least $1/4$; otherwise, the first agent would have incentive to misreport her interval as $[0,1]$, thus leading to instance $I$, in which the mechanism covers $[0,\varepsilon^2]$ with probability at least $1/4$ (since it covers $[0,\varepsilon]$ with probability at least $1/4$). Consequently, the expected max cost of the mechanism is at least $\varepsilon/4$, since it is not possible to cover more that $1-\varepsilon$ of the interval of the second agent when $[0,\varepsilon^2]$ is covered. However, the optimal max cost in $J$ is $\varepsilon^2$ by covering entirely the second agent, thus leading to an approximation ratio of at least $\varepsilon/(4\varepsilon^2) = 1/(4\varepsilon)$. 
\end{proof}

\subsection{Known but unequal interval lengths}
Another interesting variant that directly generalizes our main setting is that in which the interval lengths are known, but they are not necessarily equal. In the case of electricity supply for example, it is reasonable to assume that the government has good estimates of how much time each household requires to complete essential chores, based on verifiable information (e.g., the size of their property or the number of family members), not about their preferences on the different times of the day.

For the social cost, it is not hard to see that the vanilla median mechanism, and in fact any unweighted $k$-th ordered statistic, has an infinite approximation ratio for this case. However, we can show a linear approximation ratio by consider the {\sc Max-Length} mechanism, which places the interval at the starting position of the agent with the maximum-length interval among all agents. This mechanism is clearly truthful since the lengths are known. Without loss of generality, we will assume that the covering interval length is $1$.

\begin{theorem}
Let $\ell = \max_{i \in N} |I_i|$.
For the social cost, the approximation ratio of the {\sc Max-Length} mechanism is at most $n-1$ when $\ell \leq 1$, and at most $n$ when $\ell > 1$. 
\end{theorem}

\begin{proof}
Consider any arbitrary instance. 
Let $i$ be the agent with the max-length interval among all agents (i.e., $\ell = |I_i|$), and denote by $O$ the optimal interval. 
Without loss of generality, suppose that $O$ is at the right of $I_i$ (and thus also at the right of the covering interval $C$ chosen by the mechanism).   
Let $x \in [0,\min\{1,\ell\}]$ be the overlap between $I_i$ and $O$. 
For $j \in \{C,O\}$, let $L_j$ be the total interval length of agents at the left of $j$ that is not covered by $j$. Similarly, let $R_j$ be the total interval length of agents weakly to the right of $j$ that is not covered by $j$. 
Since $s_i \leq s_o$, we have that $L_C \leq L_O$. In case $\ell > 1$, $C$ does not cover a part $\ell_1$ of agent $i$, and thus $R_C = \max\{\ell-1,0\} + A+ R_O$, where $A$ is the total agent interval length that is covered by $O$ with the subinterval of length $\ell-x$ that does not overlap with $I_i$ from the right side. Since $i$ is the agent with the max interval length, we have $A \leq (n-1)(\ell-x)$. 
The approximation ratio is 
\begin{align*}
\frac{L_C+ R_C}{L_O+ \ell-x + R_O} 
    &\leq \frac{L_O+ \max\{\ell-1,0\} + A+ R_O}{L_O+ \ell-x + R_O} \\
    &\leq \frac{\max\{\ell-1,0\} + (n-1)(\ell-x)}{\ell-x} \\
    &\leq 
    \begin{cases}
    n-1, & \text{if } \ell \leq 1 \\
    n, & \text{if } \ell > 1 .
    \end{cases}
\end{align*}
where the second inequality follows by bounding $A$ and since $\frac{\alpha + \gamma}{\beta + \gamma} \leq \frac{\alpha}{\beta}$ for any $\alpha \geq \beta$ and $\gamma \geq 0$. In case $\ell > 1$, the last inequality follows since $\ell-1 \leq \ell-x$.  
\end{proof}

For the max cost, we show that {\sc Max-Length} achieves an approximation ratio of at most $2$. Due to Theorem~\ref{max:lower-deterministic-2}, this show that {\sc Max-Length} is essentially the best possible among all deterministic mechanisms when the lengths of the intervals are known (equal or unequal). 

\begin{theorem}
For the max cost, the approximation ratio of the {\sc Max-Length} mechanism is at most $2$. 
\end{theorem}

\begin{proof}
Consider any arbitrary instance. Let $i$ be the agent with the max-length interval among all agents. Without loss of generality, we can assume that this agent is the one with the leftmost starting point. Clearly, the approximation ratio is $1$ in case there is no overlap between the optimal covering interval $O$ and $I_i$ since then $\MC(I_i) \leq \max_{j \neq i} |I_j| \leq |I_i|$ and $\MC(O) = |I_i|$. 

Now, suppose that there is an intersection of length $x > 0$ between $I_i$ and $O$. Then, the optimal max cost is $\MC(O) = |I_i|-x$; indeed, the optimal max cost cannot be smaller than this, whereas if it was larger, then we could move the optimal covering interval towards the right to decrease it. This means that there is an interval of length at most $|I_i|-x$ after the end of $O$ that is not covered by $O$. Hence, the max cost of the mechanism is at most $2(|I_i|-x)$, that is, it is equal to the remaining length $|I_i|-x$ of $O$ that is not covered by the mechanism plus another $|I_i|-x$ that is after $O$. Consequently, the approximation ratio is at most $2$.
\end{proof}


\section{Other Open Problems and Directions}\label{sec:open}
We envision that the model we have introduced in this paper can serve as a basis for a plethora of further extensions motivated from real life scenarios. Having completely resolved the foundational version of the model, at least with respect to deterministic mechanisms, below we highlight what we consider to be some of the most prominent avenues for future work.

\paragraph{Multiple intervals.} A very meaningful extension is the one where each agent is associated with multiple intervals (say, $k_i$ intervals for agent $i$), and there are $k_c$ covering intervals to be placed (think of the the choice of several different open days at a university); those intervals could be of equal or unequal (known) length. A similar setting is one in which there is a covering budget (a total covering length $\ell_c$) which can be partitioned into intervals freely over the line. Similarly, the agents themselves could also have such interval budgets $\ell_i$; one could even impose some restrictions on the number of intervals that can be used by each agent, or by the covering budget. 

\paragraph{Different cost functions and objectives.} 
In this paper, we have considered perhaps the simplest and most intuitive cost function for the agents, namely the part of their intervals that is not covered by $C$. One could consider more complicated cost functions, e.g., functions where the cost is a convex or concave function of the proportion of the agents' interval(s) that are covered, or some most specific functions like a piecewise linear function (e.g., capturing cases where a certain amount of the interval \emph{has} to be covered for the agent to have any reduction in cost). In addition to different cost functions for the agents, one could also consider different cost functions for the aggregate objectives, such as the popular maximum cost objective, e.g., see \citep{PT09}. 

\paragraph{Utilities and social welfare.}
We could even consider (positive) utilities rather than (negative) costs. For example, in the simplest case of known and equal length intervals that we studied here, the utility of an agent would be the part of her interval that is covered, and the approximations would  be in terms of the {\em social welfare}, the total utility of the agents. It is not hard to observe that the social welfare and the social cost of a covering interval $C$ are related in this case; it holds that $\SW(C) = n - \SC(C)$. Given this, if we have a mechanism $\calM$ with a provable approximation ratio of $\rho$ in terms of the social cost, the approximation ratio of the same mechanism is at most $\frac{1}{\rho} + \frac{n(\rho-1)}{\rho \cdot \SW(\calM)}$ in terms of the social welfare. This directly gives us that the approximation ratio of the {\sc Median} mechanism is at most $n/2$, since $\rho = 2-2/n$ and $\SW(\calM) \geq 1$ (since at least one agent is completely covered). It is not hard to verify that the same arguments as in \Cref{thm:2-LB} can lead to an essentially matching lower bound for all deterministic truthful mechanisms, thus showing that the {\sc Median} mechanism is best possible even in terms of utilities and the social welfare. For more general settings, such relations between the social cost and social welfare might not exist however, and studying both of them is interesting. 

\paragraph{Obnoxious and hybrid models.}
There are also applications in which the agents might want to avoid any intersection with the covering interval. For example, when the interval corresponds to a public transportation line, the agents might want to not have any intersection with the interval since they have no interest in using the public transportation and want to avoid possible congestion or noise. On the other hand, some agents might want to have intersection with the interval in such a case as they want to use the public transportation, thus leading to interesting hybrid interval covering models. This sort of application draws parallels to similar models in the facility location literature \cite{chan2021mechanism}.

\paragraph{Higher-dimensional spaces.} One does not have to restrict attention only to intervals; a very similar setting can be defined in which each agent is associated with one or more geometric shapes on a higher dimensional space (e.g., the plane) and there is also one or more covering geometric shapes to be placed, aiming to minimize the cost of the agents as a function of the intersection with their shapes. For example, think that the covering comes from cellular antennas and every agent wants to minimize the area that they are not covered by the radius of the antenna.

\bibliographystyle{plainnat}
\bibliography{references}

\end{document}